\documentclass[aps,pra,twocolumn,superscriptaddress,nofootinbib,longbibliography]{revtex4-2}
\pdfoutput=1
\usepackage[utf8]{inputenc}
\usepackage[T1]{fontenc}
\usepackage{amsmath,amssymb,amsthm,mathtools}
\usepackage{bm}
\usepackage{hyperref}
\usepackage{graphicx}
\usepackage{booktabs}
\usepackage{tikz}
\usepackage{caption}
\usetikzlibrary{arrows.meta,calc,positioning,fit,backgrounds,shapes.geometric}

\definecolor{BCink}{HTML}{22303C}
\definecolor{BCblue}{HTML}{1F5A96}
\definecolor{BCorange}{HTML}{C2670F}
\definecolor{BCgreen}{HTML}{1B7A46}
\definecolor{BCred}{HTML}{B3261E}

\graphicspath{{figures/}}

\newtheorem{theorem}{Theorem}

\newtheorem{proposition}[theorem]{Proposition}
\newtheorem{corollary}[theorem]{Corollary}
\theoremstyle{definition}
\newtheorem{definition}[theorem]{Definition}
\newtheorem{remark}[theorem]{Remark}

\newsavebox{\tabQ}
\newcommand{\Tr}{\operatorname{Tr}}
\newcommand{\CPO}{\mathrm{CPO}}
\newcommand{\CPTP}{\mathrm{CPTP}}
\newcommand{\SC}{\mathrm{SC}}
\newcommand{\sgn}{\operatorname{sgn}}
\newcommand{\diag}{\operatorname{diag}}

\newcommand{\rank}{\operatorname{rank}}
\newcommand{\range}{\operatorname{range}}

\begin{document}

\title{Pairwise detection, chiral conversion: Resource theory of quantum multistates}

\author{Yuan-Hong Tao}
\author{Yu-Xin Wu}
\affiliation{College of Science, Zhejiang University of Science and Technology, Hangzhou, Zhejiang 310023, China}

\author{Bing Xie}
\email{E-mail: xiebingjiangxi2023@163.com}
\affiliation{School of Mathematical Sciences, Hangzhou Dianzi
University, Hangzhou, Zhejiang 310018, China}
\author{Lin Zhang}
\email{Contact author: godyalin@163.com}
\affiliation{School of Mathematical Sciences, Hangzhou Dianzi
University, Hangzhou, Zhejiang 310018, China}

\date{\today}

\begin{abstract}
Basis-independent coherence of a quantum multistate is fixed by whether its constituent states commute, and for qubits it is detected by pairwise state overlaps alone. What governs its deterministic transformation has been open. We give the exact resource-conversion theory of quantum multistates under common commutativity-preserving operations (CPOs). For qubits, pairwise Bargmann invariants determine the singular values of the unique candidate channel, while a single third-order Bargmann phase determines its orientation. Complete positivity then yields a necessary and sufficient criterion for exact conversion. The criterion identifies a chirality-sensitive shell $\mathcal C_{\rm chir}$, within which targets with identical pairwise geometry can differ in convertibility solely by the sign of their cubic invariant. This yields an order separation, $k_{\rm detect}=2<3=k_{\rm convert}$, for qubit multistates of every size $n\geqslant3$. The chirality-sensitive shell occupies exactly one third of the singular-value simplex, while mirror reflection obstructs two thirds of the admissible orientation-preserving transformations by volume. For $d\geqslant3$, by contrast, the structure is rigid: every resourceful conversion is realized by either an isotropic depolarizing or a transpose-depolarizing channel. Convertibility is then decided by centered Bargmann invariants of degree at most $d^2$.
\end{abstract}

\maketitle

\section{Introduction}\label{sec:intro}

Quantum resource theories quantify non-classical features of quantum states such as entanglement, coherence, and asymmetry~\cite{Horodecki2007,Piani2014,Gour2018}. In the usual resource theory of coherence, coherence is defined relative to a designated reference basis~\cite{Baumgratz2014}. Distributed quantum protocols often have no shared reference frame, which makes basis-independent formulations the relevant ones. The term has two inequivalent meanings, referring either to a unitary-invariant quantity of a single state \cite{Jin2024} or to the relational notion used here, set coherence or multistate compatibility \cite{Li2017,Li2019,Mao2019,Designolle2021}. 
An ensemble is set-coherent exactly when its density operators cannot be simultaneously diagonalized, that is, when they fail to commute. For qubits, this is already determined by second-order pairwise data. In dimensions three and above, however, strictly higher-order Bargmann data are required ~\cite{Wang2026}.

This raises the question we address: \emph{if the presence of set coherence in a qubit ensemble is determined by pairwise overlaps alone, are pairwise invariants also enough to decide the transformation of one coherent ensemble into another?}

\begin{figure*}[t]
\includegraphics[width=\textwidth]{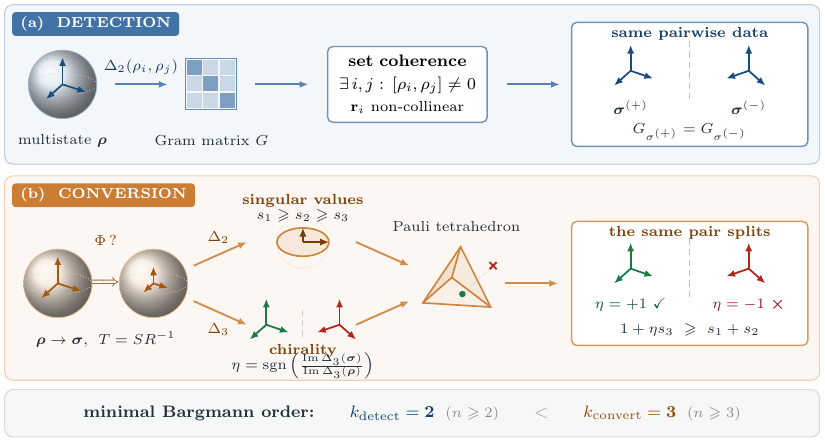}
\captionsetup{width=\textwidth}
\caption{(Color Online) Detection and conversion of set coherence read the multistate at different Bargmann orders. (a)~Detection uses only the pairwise invariants $\Delta_2(\rho_i,\rho_j)$, which determine the Gram matrix $G$. Set coherence is detected by $\operatorname{rank}(G)\geqslant2$. Mirror reflection preserves $G$, so a target and its mirror image have identical order-2 data. (b)~Conversion fixes the candidate channel $T=SR^{-1}$, whose singular values $s_1\geqslant s_2\geqslant s_3$ are determined by the order-2 invariants. Its orientation is instead fixed by the sign
$\eta=\operatorname{sgn}\Big(\frac{\operatorname{Im}\Delta_3(\boldsymbol{\sigma})}{\operatorname{Im}\Delta_3(\boldsymbol{\rho})}\Big)$ of two third-order invariants. Complete positivity places the signed spectrum $(s_1,s_2,\eta s_3)$ inside the Pauli tetrahedron, so the mirror pair that detection cannot separate splits into a realizable and an obstructed transformation. Detection is therefore settled at order $2$ for every $n \geqslant 2$, conversion at order $3$ for every $n \geqslant 3$.}\label{fig:paradigm_shift}
\end{figure*}

They are not. Take three pure qubit states whose Bloch vectors are mutually orthogonal and scale them by a contraction factor $c \in (\frac13, 1]$. The original configuration converts deterministically into its scaled copy, while its mirror image is out of reach of every commutativity-preserving operation. Yet the two targets have the same pairwise overlaps, fidelities, and purities. What pairwise data fail to capture, but complete positivity constrains, is orientation, given by the sign of a third-order Bargmann phase.

Bargmann invariants~\cite{Fuchs2003,Galvao2020,Oszmaniec2024} have been studied recently in quantum kinematics, multi-mode interferometry, and metric reconstruction~\cite{Li2025,Wagner2026a,Wagner2026b,Wang2026,Guo2026}, mostly as static geometric descriptors. Their role in synthesising a channel has had less attention. How relational coherence evolves under a channel is also tied to the trade-off between coherence and mixedness in open-system dynamics~\cite{Che2023}.

Here, we construct the exact resource-conversion theory of quantum multistates under common commutativity-preserving operations (CPOs). Our main contributions are:
\begin{itemize}
    \item \emph{Order separation ($k_{\mathrm{detect}} = 2 < 3 = k_{\mathrm{convert}}$).} Detection requires only second-order Bargmann invariants, whereas exact conversion requires third-order invariants (Fig.~\ref{fig:paradigm_shift}). The third-order requirement is tight for every ensemble size \(n\ge3\). Moreover, degree three is globally complete, as it determines a qubit multistate up to a common unitary. Therefore, no qubit conversion problem can require higher-order invariants. 
    \item \emph{The chirality-sensitive shell $\mathcal{C}_{\mathrm{chir}}$.} We derive a necessary and sufficient geometric criterion for exact qubit conversion based on the signed singular values of the candidate channel. This criterion characterizes \(C_{\rm chir}\) as the region where complete positivity distinguishes the two orientations. The shell occupies exactly one third of the singular-value simplex. By volume, two thirds of the transformations admissible in the orientation-preserving branch become non-convertible upon reflection. 
    \item \emph{The dimensional anomaly.} For $d \geqslant 3$ the anisotropic contractions available to a qubit disappear: every CPO with non-commutative range is an isotropic depolarizing or transpose-depolarizing channel~\cite{Hu2012,Guo2013}. On that classification we build a decision procedure, reading the branch and the contraction parameter off centered Bargmann invariants of degree at most $d^2$.
\end{itemize}

\section{Framework: Multistates and Commutativity-Preserving Operations}
\label{sec:framework}

Let $\mathcal{H} \cong \mathbb{C}^d$ be a finite-dimensional Hilbert space, and let $\mathcal{D}(\mathbb{C}^d)$ denote the compact convex set of density operators on $\mathcal{H}$.

An $n$-multistate is an ordered tuple of density operators $\boldsymbol{\rho} = (\rho_1, \ldots, \rho_n) \in \mathcal{D}(\mathbb{C}^d)^n$. It is free (set-incoherent) if there exists an orthonormal basis in which all its constituent states are simultaneously diagonal. The set of all such multistates is denoted by \(\mathcal F_{n,d}\). Equivalently,
\begin{equation*}
\boldsymbol{\rho} \in \mathcal{F}_{n,d} \iff [\rho_i, \rho_j] = 0, \quad \forall i, j \in \{1, \ldots, n\}.
\end{equation*}
Otherwise, \(\boldsymbol{\rho}\) is resourceful (set-coherent), meaning that at least one pair of its constituent states fails to commute.

We consider physical transformations of a multistate described by a common completely positive trace-preserving (CPTP) map $\Phi:\mathcal B(\mathcal H)\to\mathcal B(\mathcal H)$. Its action on the multistate is
$\Phi^{\times n}(\boldsymbol{\rho})
=\bigl(\Phi(\rho_1),\ldots,\Phi(\rho_n)\bigr)$.
The maximal resource-nongenerating class consists of those CPTP maps that take free
multistates to free multistates, namely
\[
\mathfrak F_{n,d}
=
\left\{
\Phi\in\mathrm{CPTP}
\,\middle|\,
\Phi^{\times n}(\mathcal F_{n,d})
\subseteq
\mathcal F_{n,d}
\right\}.
\]

Since freeness is characterized by pairwise commutativity,
we introduce the commutativity-preserving operations (CPOS)
\[
\CPO_d :=
\left\{
\Phi\in\CPTP \;\middle|\;
\begin{array}{l}
    [\rho,\tau]=0 \Longrightarrow
    [\Phi(\rho),\Phi(\tau)]=0\\
    \text{for all }\rho,\tau\in\mathcal D(\mathbb C^d)
\end{array}
\right\}.
\]
For every $n\geqslant2$, we have
$\widetilde{\mathfrak F}_{n,d}=\CPO_d$.
One inclusion is immediate, and the reverse follows by
applying the free-preservation condition to
$(\rho,\tau,\rho,\ldots,\rho)$ for any commuting pair
$\rho,\tau$.

Exact conversion of $\boldsymbol{\rho}$ into
$\boldsymbol{\sigma}$ by a common free channel,
denoted by $\boldsymbol{\rho}\succeq_{\SC}\boldsymbol{\sigma}$,
is equivalent to the existence of $\Phi\in\CPO_d$ satisfying
$\Phi(\rho_i)=\sigma_i$ for all $i$.
This relation is a preorder because $\CPO_d$ contains
the identity channel and is closed under composition.

For unconstrained quantum channels, exact conversion
between prescribed states is a settled question.
The Alberti--Uhlmann criterion and its extensions decide
it for dichotomies and beyond~\cite{Dall2020}, and quantum
majorization decides it for arbitrary families~\cite{Gour2018}.
Here, the channel is restricted to the free class $\CPO_d$,
and this restriction makes orientation relevant to
convertibility.

When the target multistate is free, conversion reduces to a POVM feasibility problem, regardless of the dimension
or the source multistate.

\begin{proposition}[POVM criterion for conversion to free multistates]
\label{prop:free_target}
Let $\boldsymbol{\rho} \in \mathcal{D}(\mathbb{C}^d)^n$ and $\boldsymbol{\sigma} \in \mathcal{F}_{n,d}$, with $\sigma_i = \sum_a q_{a|i}\,|a\rangle\langle a|$ in a common eigenbasis $\{|a\rangle\}_{a=1}^d$. Then $\boldsymbol{\rho} \succeq_{\SC} \boldsymbol{\sigma}$ if and only if there is a POVM $\{M_a\}_{a=1}^d$ on $\mathcal{H}$ with
\begin{equation}
\label{eq:free_target}
\Tr(M_a \rho_i) = q_{a|i} \qquad \text{for all } a, i.
\end{equation}
\end{proposition}

\begin{proof}
Given such a POVM, the measure-and-prepare map $\Phi(X) = \sum_a \Tr(M_a X)\,|a\rangle\langle a|$ is completely positive and trace-preserving, and its image lies in the abelian algebra of operators diagonal in $\{|a\rangle\}$. It therefore sends every multistate to a free one, so $\Phi \in \mathfrak{F}_{n,d} = \CPO_d$, and $\Phi(\rho_i) = \sigma_i$ by~\eqref{eq:free_target}.

Conversely, let $\Phi \in \CPO_d$ satisfy $\Phi(\rho_i) = \sigma_i$. Set $M_a = \Phi^\dagger(|a\rangle\langle a|)$, where the adjoint is taken in the Hilbert--Schmidt inner product. Since $\Phi$ is completely positive, its adjoint
$\Phi^\dagger$ is also completely positive, giving
$M_a=\Phi^\dagger(|a\rangle\langle a|)\geqslant0$.
Moreover, trace preservation of $\Phi$ implies
$\Phi^\dagger(\mathbb I)=\mathbb I$, so
\[
\sum_a M_a = \Phi^\dagger\!\left(\sum_a|a\rangle\langle a|\right) = \Phi^\dagger(\mathbb I) = \mathbb I.
\]
Thus, $\{M_a\}$ is a POVM. Finally, $\Tr(M_a \rho_i) = \Tr(|a\rangle\langle a|\,\Phi(\rho_i)) = \langle a|\sigma_i|a\rangle = q_{a|i}$.
\end{proof}

The proposition reduces conversion to a free target to a POVM feasibility problem. We now turn to resourceful targets. A $\CPO_d$ realizing such a target cannot have commuting range, since otherwise its outputs would form a free multistate.


\section{Exact qubit conversion and the chirality barrier}\label{sec:qubit_conversion}

\subsection{Bloch space representation and qubit CPOs}

For a single qubit ($d = 2$), any density operator decomposes in the Pauli basis $\boldsymbol{\sigma} = (\sigma_x, \sigma_y, \sigma_z)$ as $\rho = \frac{1}{2}(\mathbb{I} + \mathbf{r} \cdot \boldsymbol{\sigma})$, where $\mathbf{r} \in \mathbb{R}^3$ satisfies $\|\mathbf{r}\| \leqslant 1$. Two states $\rho, \tau$ commute if and only if their Bloch vectors are collinear. Set coherence of $\boldsymbol{\rho} = (\rho_1, \ldots, \rho_n)$ is therefore equivalent to non-collinearity of its Bloch vectors $\{\mathbf{r}_i\}_{i=1}^n$.

\begin{theorem}[Classification of qubit CPOs~\cite{Hu2012,Guo2013}]
\label{thm:qubit_cpo_classification}
A qubit CPTP channel $\Phi$ is commutativity-preserving if and only if it satisfies at least one of the following two conditions:
\begin{enumerate}
    \item $\Phi$ is unital, $\Phi(\mathbb{I}) = \mathbb{I}$. In Bloch coordinates, it acts as a real linear map: $\mathbf{r} \mapsto T \mathbf{r}$, with $T \in \mathbb{R}^{3\times3}$.
    \item $\Phi$ has commuting range: $[\Phi(\rho), \Phi(\tau)] = 0$ for all $\rho, \tau \in \mathcal{D}(\mathbb{C}^2)$. In the affine Bloch form $\mathbf{r} \mapsto T\mathbf{r} + \mathbf{c}$ this says that $\range(T)$ and $\mathbf{c}$ lie on one fixed line, confining every output to a single commutative line in Bloch space.
\end{enumerate}
\end{theorem}

Because a channel with commuting range maps every multistate to a free (set-incoherent) one, it cannot realize a resourceful target. Hence, for qubits, Theorem~\ref{thm:qubit_cpo_classification} implies that any CPO converting a resourceful multistate $\boldsymbol{\rho}$ into another resourceful multistate $\boldsymbol{\sigma}$ must be unital. Consequently, any realizing channel acts on the Bloch vectors as $\mathbf{s}_i = T\mathbf{r}_i$ for some real matrix $T\in\mathbb R^{3\times3}$ and all $i\in\{1,\ldots,n\}$.

\begin{figure*}[t]
\includegraphics[width=\textwidth]{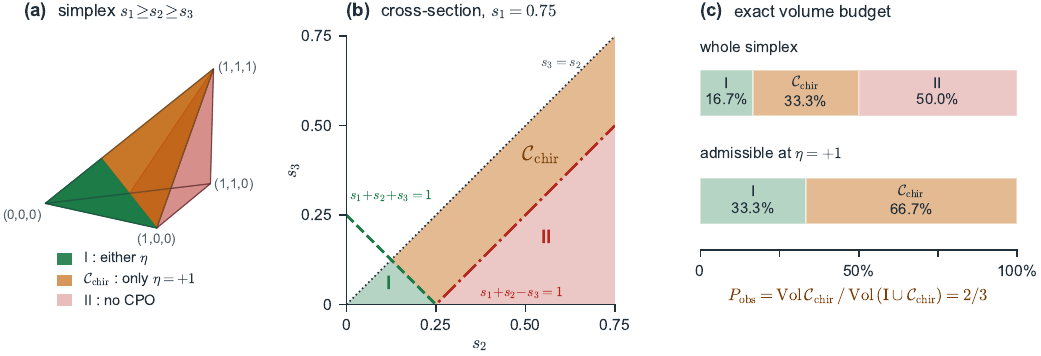}
\caption{The chirality-sensitive shell $\mathcal{C}_{\mathrm{chir}}$
    inside the ordered singular-value simplex.
    (a)~The planes $s_1+s_2+s_3=1$ and $s_1+s_2-s_3=1$
    partition the simplex into three regions.
    Region~I is admissible for both orientations,
    the shell $\mathcal{C}_{\mathrm{chir}}$ (Region~III) is admissible only for $\eta=+1$, and Region~II is inadmissible for either orientation.
    (b)~The cross-section at $s_1=0.75$ shows the shell between the two boundary lines, bounded above by $s_3=s_2$.
    (c)~Exact volume fractions of the three regions. The shell occupies two thirds of the volume admissible for $\eta=+1$. Throughout this shell, reflection reverses the orientation   and violates complete positivity.}
\label{fig:chirality_shell}
\end{figure*}

\subsection{Bargmann invariants and candidate channel synthesis}

Let $\boldsymbol{\rho} = (\rho_1, \rho_2, \rho_3)$ have non-singular Bloch matrix $R = [\mathbf{r}_1, \mathbf{r}_2, \mathbf{r}_3]$, and let $\boldsymbol{\sigma} = (\sigma_1, \sigma_2, \sigma_3)$ be any triple with Bloch matrix $S = [\mathbf{s}_1, \mathbf{s}_2, \mathbf{s}_3]$. Let $\Phi_T$ denote the unital map whose Bloch action is $T$. The candidate map $\Phi_T$ is uniquely fixed by $T = S R^{-1}$, whether or not $S$ is invertible. Convertibility then reduces to whether this candidate unital map is completely positive, and hence a valid CPO.

We define the $k$-th order Bargmann invariant for a sequence of quantum states as
\begin{equation*}
\Delta_k(\rho_1, \ldots, \rho_k) = \Tr(\rho_1 \rho_2 \cdots \rho_k).
\end{equation*}
For qubits, the $k=2$ case provides the Hilbert-Schmidt inner product
\begin{equation*}
\Tr(\rho_i \rho_j) = \tfrac{1}{2}(1 + \mathbf{r}_i \cdot \mathbf{r}_j) \implies \mathbf{r}_i \cdot \mathbf{r}_j = 2\Tr(\rho_i \rho_j) - 1.
\end{equation*}
Thus, the Gram matrices $G_{\boldsymbol{\rho}} = R^\mathsf{T} R$ and $G_{\boldsymbol{\sigma}} = S^\mathsf{T} S$ are expressed directly through pairwise Bargmann invariants by
\begin{align*}
[G_{\boldsymbol{\rho}}]_{ij} &= 2\Delta_2(\rho_i, \rho_j) - 1, \\
[G_{\boldsymbol{\sigma}}]_{ij} &= 2\Delta_2(\sigma_i, \sigma_j) - 1.
\end{align*}

\begin{theorem}[Spectral and orientation identification]
\label{thm:gram_svd}
Let $T = S R^{-1}$. Its singular values $s_1 \geqslant s_2 \geqslant s_3 \geqslant 0$, and its determinant sign $\eta = \sgn\det(T)$ when $T$ is invertible, are determined by relational invariants:
\begin{enumerate}
    \item The squared singular values $\{s_k^2\}_{k=1}^3$ are the generalized eigenvalues of the Bargmann Gram pencil, given by the roots of
    \begin{equation*}
    \det(G_{\boldsymbol{\sigma}} - \mu G_{\boldsymbol{\rho}}) = 0.
    \end{equation*}
    \item If $\det (S) \neq 0$, the orientation $\eta \in \{+1, -1\}$ is the ratio of the imaginary parts of the third-order Bargmann invariants, given by
    \begin{equation*}
    \eta = \sgn\left( \frac{\operatorname{Im}\Delta_3(\sigma_1, \sigma_2, \sigma_3)}{\operatorname{Im}\Delta_3(\rho_1, \rho_2, \rho_3)} \right).
    \end{equation*}
    If $\det (S) = 0$, then $s_3 = 0$, this ratio vanishes, and we set $\eta = +1$ by convention.
\end{enumerate}
\end{theorem}

\begin{proof}
Since \(S=TR\), we have
$G_{\boldsymbol{\sigma}}=S^{T}S=R^{T}T^{T}TR$. Therefore, $G_{\boldsymbol{\sigma}}-\mu G_{\boldsymbol{\rho}}
=R^{T}(T^{T}T-\mu I)R$, and hence
\[
\det(G_{\boldsymbol{\sigma}}-\mu G_{\boldsymbol{\rho}})
=\det(R)^2\det(T^{T}T-\mu I).
\]Since \(R\) is nonsingular, the roots of $\det(G_{\boldsymbol{\sigma}}-\mu G_{\boldsymbol{\rho}})=0$ are precisely the eigenvalues of \(T^{T}T\), namely $s_1^2,\ s_2^2,\ s_3^2$. This proves item 1.

For item 2, the identity for Pauli matrices yields
\begin{equation*}
\Tr(\rho_1 \rho_2 \rho_3) = \frac{1}{4} + \frac{1}{4}\sum_{i<j}\mathbf{r}_i\cdot\mathbf{r}_j + \frac{\mathrm{i}}{4}(\mathbf{r}_1 \times \mathbf{r}_2) \cdot \mathbf{r}_3.
\end{equation*}
Therefore, $\operatorname{Im}\Delta_3(\rho_1, \rho_2, \rho_3) = \frac{1}{4}\det(R)$. Since $S = T R$, we have $\det(S) = \det(T)\det(R)$, and taking signs yields the result.
\end{proof}

\begin{figure*}[t]
\includegraphics[width=\textwidth]{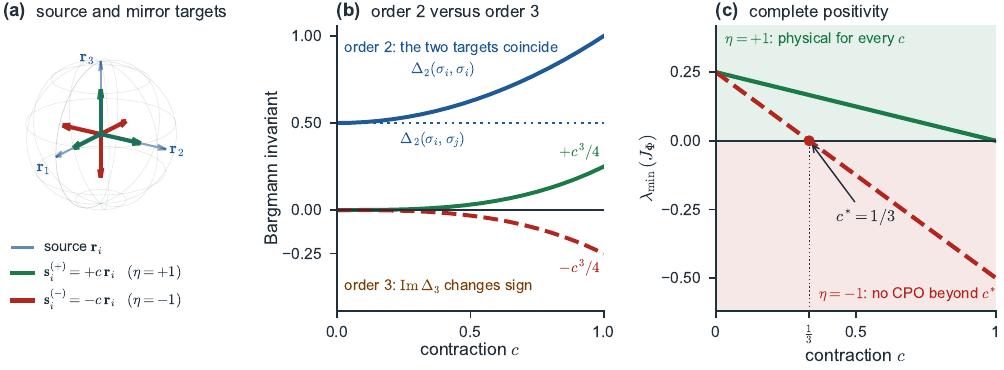}
\caption{The mirror counterexample. (a)~Three orthogonal Bloch vectors $\mathbf{r}_i$ are mapped either onto the contracted copy $\mathbf{s}^{(+)}_i=c\,\mathbf{r}_i$ or onto its reflection $\mathbf{s}^{(-)}_i=-c\,\mathbf{r}_i$. (b)~Every order-2 invariant of the two targets coincides: the cross terms stay at $\Delta_2(\sigma_i,\sigma_j)=\frac12$ and the purities at $\frac{1+c^2}2$. The order-3 invariant separates them, $\operatorname{Im}\Delta_3=\pm \frac{c^3}4$. (c)~The minimum Choi eigenvalue of the corresponding candidate channel. The orientation-preserving branch stays positive for every $c$, whereas the reflected branch crosses zero at $c^*=\frac13$ and is obstructed beyond it.}
\label{fig:mirror_counterexample}
\end{figure*}

\subsection{Exact conversion criterion and the chirality shell}

Any real matrix $T \in \mathbb{R}^{3\times3}$ admits a signed singular value decomposition~\cite{Fujiwara1999}
\begin{equation*}
T = O_1 \diag(s_1, s_2, \eta s_3) O_2^{\mathsf{T}}, \quad O_1, O_2 \in \mathrm{SO}(3),
\end{equation*}
where $s_1 \geqslant s_2 \geqslant s_3 \geqslant 0$ and $\eta = \sgn\det(T) \in \{+1, -1\}$. Because $\mathrm{SO}(3)$ corresponds to single-qubit unitary operations, there exist unitary channels $\mathcal U_1,~\mathcal U_2$ such that $\Phi_T=\mathcal U_1\circ\Phi_{\diag(\boldsymbol{\lambda})}\circ \mathcal U_2$, $\boldsymbol{\lambda} = (\lambda_1, \lambda_2, \lambda_3) = (s_1, s_2, \eta s_3)$. Therefore, the complete positivity of $\Phi_T$ is entirely determined by the signed singular values of $T$.

The normalized Choi matrix of the diagonal unital map $\Phi_{\diag(\boldsymbol{\lambda})}$ is diagonal in the Bell basis~\cite{King2001,Ruskai2002}, with eigenvalues
\begin{equation}
\label{eq:choi_spectrum}
\begin{aligned}
p_0 &= \tfrac{1}{4}(1 + \lambda_1 + \lambda_2 + \lambda_3), \\
p_1 &= \tfrac{1}{4}(1 + \lambda_1 - \lambda_2 - \lambda_3), \\
p_2 &= \tfrac{1}{4}(1 - \lambda_1 + \lambda_2 - \lambda_3), \\
p_3 &= \tfrac{1}{4}(1 - \lambda_1 - \lambda_2 + \lambda_3).
\end{aligned}
\end{equation}
Thus, the map is completely positive if and only if $p_i\geq 0$ for all $i$, i.e. $\boldsymbol{\lambda}$ lies within the Pauli tetrahedron defined by
\begin{equation}
\label{eq:pauli_tetrahedron}
1 + \lambda_3 \geqslant |\lambda_1 + \lambda_2|
\quad \text{and} \quad
1 - \lambda_3 \geqslant |\lambda_1 - \lambda_2|.
\end{equation}
Substituting $\boldsymbol{\lambda}=(s_1,s_2,\eta s_3)$ and using $s_1\geqslant s_2\geqslant s_3\geqslant0$, we obtain the following exact conversion criterion.

\begin{theorem}[Exact qubit state conversion]
\label{thm:qubit_exact}
Let $\boldsymbol{\rho}$ be a qubit triple with non-singular Bloch matrix $R$, and let $\boldsymbol{\sigma}$ be any resourceful triple. Write $s_1 \geqslant s_2 \geqslant s_3 \geqslant 0$ and $\eta$ for the singular values and orientation of $T = S R^{-1}$. Then $\boldsymbol{\rho} \succeq_{\SC} \boldsymbol{\sigma}$ if and only if
\begin{equation}
\label{eq:qubit_master}
1 + \eta s_3 \geqslant s_1 + s_2 \quad \text{and} \quad 1 - \eta s_3 \geqslant s_1 - s_2.
\end{equation}
\end{theorem}

Within the ordered singular-value simplex
$\mathcal{S}=\{(s_1,s_2,s_3):1\geqslant s_1\geqslant s_2
\geqslant s_3\geqslant0\}$, condition~\eqref{eq:qubit_master}
partitions $\mathcal{S}$ into three geometrically distinct regions
according to whether complete positivity holds for each orientation
(Fig.~\ref{fig:chirality_shell}).
In Region~I, where $s_1+s_2+s_3\leqslant1$, the candidate maps
are completely positive for both orientations.
In Region~II, where $s_1+s_2-s_3>1$, neither candidate map
is completely positive, so neither defines a quantum channel.

The remaining region,
\[
\mathcal{C}_{\mathrm{chir}}
=
\left\{(s_1,s_2,s_3)\in\mathcal{S}:
s_1+s_2-s_3\leqslant1<s_1+s_2+s_3
\right\},
\]
is the chirality-sensitive shell (Region~III).
Here, the candidate map with $\eta=+1$ is completely positive,
whereas its mirror-reflected counterpart with $\eta=-1$ is not.
Thus, within this shell, orientation alone distinguishes
realizable from forbidden conversions with identical
singular values.

\subsection{The mirror counterexample and the chirality threshold}
\label{sec:mirror}

The operational necessity of the third-order invariant is demonstrated by an explicit family of state configurations. Let the source multistate consist of three pure qubit states
whose Bloch vectors are
\begin{equation*}
\mathbf{r}_1 = (1, 0, 0)^{\mathsf{T}}, \quad \mathbf{r}_2 = (0, 1, 0)^{\mathsf{T}}, \quad \mathbf{r}_3 = (0, 0, 1)^{\mathsf{T}}.
\end{equation*}
For a contraction factor $c\in[0,1]$, define two target multistates $\boldsymbol{\sigma}^{(+)}$ and $\boldsymbol{\sigma}^{(-)}$ by
$$\mathbf{s}_i^{(+)}=c\mathbf{r}_i,
\qquad
\mathbf{s}_i^{(-)}=-c\mathbf{r}_i,
\qquad i=1,2,3.$$
Both targets share identical pairwise Gram matrices
\begin{equation*}
G_{\boldsymbol{\sigma}^{(+)}} = G_{\boldsymbol{\sigma}^{(-)}} = c^2 \mathbb{I}_3.
\end{equation*}
The two targets have identical pairwise fidelities, purities, and second-order Bargmann invariants, so these data alone cannot distinguish them. Since $T_\pm=\pm c\mathbb{I}_3$, both candidate maps $\Phi_{T_{\pm}}$ have the same singular values $s_1=s_2=s_3=c$.

Substituting $\boldsymbol{\lambda}=(c,c,\eta c)$ into \eqref{eq:choi_spectrum} gives the minimum Choi eigenvalue
\begin{equation*}
\lambda_{\min}(J_{\Phi_{\diag(\boldsymbol{\lambda})}}) = \begin{cases}
\frac{1 - c}{4}, & \eta = +1, \\
\frac{1 - 3c}{4}, & \eta = -1.
\end{cases}
\end{equation*}
As illustrated in Fig.~\ref{fig:mirror_counterexample},
the candidate map for $\eta=+1$ is completely
positive for every $c\in[0,1]$, so
$\boldsymbol{\rho}\succeq_{\SC}\boldsymbol{\sigma}^{(+)}$ throughout this interval.
For $\eta=-1$, the minimum Choi eigenvalue
vanishes at the critical contraction $c^*=\frac{1}{3}$ and becomes negative for $c>c^*$.
Thus, no CPO can convert $\boldsymbol{\rho}$ into
$\boldsymbol{\sigma}^{(-)}$ for $c\in(1/3,1]$.

For any fixed $c\in(1/3,1]$, the two targets have identical second-order Bargmann data, yet only $\boldsymbol{\sigma}^{(+)}$ is reachable from $\boldsymbol{\rho}$ under CPOs. Their opposite orientations are encoded in
$$\operatorname{Im}\Delta_3(\boldsymbol{\sigma}^{(\pm)})=\pm\frac{c^3}{4}.$$
Thus, distinguishing these two conversion outcomes requires Bargmann data of at least third order.

\begin{remark}[Source of the obstruction]
\label{rem:nonunital}
The barrier at $c^*$ is resource-theoretic, not a failure of complete positivity alone. Drop the free-operation constraint and consider the affine Bloch map $\mathbf{r} \mapsto A\mathbf{r} + \mathbf{b}$ with
\begin{equation*}
A = -\tfrac{1}{2} \mathbb{I}_3 + \tfrac{1}{4}\mathbf{e}\mathbf{e}^\mathsf{T}, \qquad \mathbf{b} = -\tfrac{1}{4}\mathbf{e},
\end{equation*}
where $\mathbf{e} = (1,1,1)^\mathsf{T}$. It sends $\mathbf{r}_i \mapsto -\tfrac{1}{2}\mathbf{r}_i$ for each $i$, and its Choi spectrum $\{(5 \pm \sqrt{19})/16,\ (3 \pm \sqrt{3})/16\}$ is strictly positive, so it is a channel reaching $\boldsymbol{\sigma}^{(-)}$ at $c = 1/2 > c^*$. It is not unital, since $\mathbf{b} \ne \mathbf{0}$, and its range does not commute, since $A$ is invertible. By Theorem~\ref{thm:qubit_cpo_classification}, it is therefore not a CPO. What fails beyond $c^*$ is the conjunction of complete positivity with resource non-generation, and it is that conjunction the orientation detects.
\end{remark}

\begin{definition}[Bargmann order of a task]
\label{def:orders}
Let the data available about a pair of $n$-multistates be the Bargmann invariants of $\boldsymbol{\rho}$ and of $\boldsymbol{\sigma}$ up to some degree, each taken within its own multistate, with no cross invariants mixing $\rho_i$ and $\sigma_j$. Then $k_{\mathrm{detect}}(n)$ is the least degree at which those data decide whether an $n$-multistate is set-coherent, and $k_{\mathrm{convert}}(n)$ the least degree at which they decide whether $\boldsymbol{\rho} \succeq_{\SC} \boldsymbol{\sigma}$, in both cases for every instance of that size.
\end{definition}

%
%

\begin{remark}
\label{rem:scoping}
The exclusion of cross invariants in Definition~\ref{def:orders} is not cosmetic. If the joint invariants of the concatenated $2n$-tuple are admitted as well, the degree-two overlaps assemble the cross-Gram entries $2\Delta_2(\rho_i, \sigma_j) - 1$. Choose any spanning triple of source indices $J$ and write $R_J, S_J \in \mathbb{R}^{3\times3}$ for the corresponding Bloch blocks. Then
\begin{equation}
\label{eq:joint_order_two}
\det(R_J^T S_J) = \det(T)\,\det(R_J^T R_J),
\end{equation}
so the orientation is already fixed at degree two and the separation collapses. Convertibility is invariant under independent unitary conjugation
of the source and target. If $\Phi$ realizes $\boldsymbol{\rho}\succeq_{\SC}\boldsymbol{\sigma}$, then $\operatorname{Ad}_V\circ\Phi\circ\operatorname{Ad}_{U^\dagger}$
realizes $U\boldsymbol{\rho}U^\dagger \succeq_{\SC} V\boldsymbol{\sigma}V^\dagger$. This channel is again a CPO, so convertibility is constant on the orbits of the independent $\mathrm{U}(d)\times\mathrm{U}(d)$ action.
\end{remark}

Second- and third-order Bargmann invariants not only suffice for criterion~\eqref{eq:qubit_master} but also determine a qubit multistate up to a common unitary.

\begin{theorem}[Order 3 determines a qubit multistate]
\label{thm:global_order}
Let $\boldsymbol{\rho}, \boldsymbol{\rho}' \in \mathcal{D}(\mathbb{C}^2)^n$ satisfy $\Delta_2(\rho_i, \rho_j) = \Delta_2(\rho'_i, \rho'_j)$ and $\Delta_3(\rho_i, \rho_j, \rho_k) = \Delta_3(\rho'_i, \rho'_j, \rho'_k)$ for all indices. Then there is a single unitary $U$ with $\rho'_i = U \rho_i U^\dagger$ for every $i$.
\end{theorem}

\begin{proof}
Collect the Bloch vectors into $R = [\mathbf{r}_1, \ldots, \mathbf{r}_n]$ and $R' = [\mathbf{r}'_1, \ldots, \mathbf{r}'_n]$ in $\mathbb{R}^{3 \times n}$. Degree-two equality and $\mathbf{r}_i \cdot \mathbf{r}_j = 2\Delta_2(\rho_i,\rho_j) - 1$ give $R^T R = R'^T R'$, so $R' = O R$ for some $O \in \mathrm{O}(3)$.

If $\rank (R) \leqslant 2$, pick a unit vector $\mathbf{n} \bot \range(R)$ and set $Q = \mathbb{I}_3 - 2\,\mathbf{n}\mathbf{n}^\mathsf{T}$. Then $QR = R$ and $\det Q = -1$, so replacing $O$ by $OQ$ leaves $R' = OR$ intact and reverses the sign of the determinant; $O$ may therefore be taken in $\mathrm{SO}(3)$.

If $\rank R = 3$, some triple of columns is linearly independent, say $(i,j,k)$. By the identity in the proof of Theorem~\ref{thm:gram_svd}, $\operatorname{Im}\Delta_3(\rho_i,\rho_j,\rho_k) = \tfrac{1}{4}\det[\mathbf{r}_i, \mathbf{r}_j, \mathbf{r}_k] \ne 0$, and the same holds for $\boldsymbol{\rho}'$ with the determinant multiplied by $\det O$. Equality of the degree-three invariants forces $\det (O) = +1$.

In both cases $O \in \mathrm{SO}(3)$, which is the Bloch action of conjugation by some $U \in \mathrm{SU}(2)$.
\end{proof}

Bargmann invariants play the same role in deciding local unitary equivalence of multipartite states~\cite{Zhang2025}. Theorem~\ref{thm:global_order} is the multistate counterpart, with the orientation supplied by a single cubic phase.

That degree three closes the invariant theory of $2 \times 2$ matrix tuples is classical: Sibirskii showed that the simultaneous-conjugation invariants of such tuples are generated in degree at most three~\cite{Sibirskii1968}. What Theorem~\ref{thm:global_order} adds is the identification of the missing cubic datum as chirality, in a form that reads directly off the Bloch geometry, and it is that datum which the next corollary shows to be operationally necessary.

Together with the independent-unitary invariance of convertibility noted in Remark~\ref{rem:scoping}, Theorem~\ref{thm:global_order} allows us to determine the Bargmann orders for qubit multistates of arbitrary size and source rank.

\begin{corollary}[Global order separation]
\label{cor:global_order}
For qubit multistates, $k_{\mathrm{detect}}(2)
=k_{\mathrm{convert}}(2)=2$, whereas for every $n\geqslant3$,
\[
k_{\mathrm{detect}}(n)=2<3=k_{\mathrm{convert}}(n).
\]
\end{corollary}

\begin{proof}
Convertibility is constant on the orbits of the independent unitary action of Remark~\ref{rem:scoping}, so it is a function of those orbits alone. By Theorem~\ref{thm:global_order} the invariants of degree at most three separate the orbits, hence fix convertibility for every qubit multistate, spanning or not; thus $k_{\mathrm{convert}}(n) \leqslant3$.

Order one never suffices. Take the free source $\boldsymbol{\rho} = (|0\rangle\langle 0|, |1\rangle\langle 1|)$ and compare the target $\boldsymbol{\sigma}^{(A)} = \boldsymbol{\rho}$, reached by the identity channel, with $\boldsymbol{\sigma}^{(B)} = (|0\rangle\langle 0|, |{+}\rangle\langle{+}|)$, which is resourceful and so out of reach of any CPO acting on a free source. Every state has $\Delta_1 = \Tr\rho_i = 1$, so degree-one data cannot separate the two instances; padding both sides with copies of $\frac{\mathbb{I}}2$ carries the example to any size. Hence $k_{\mathrm{convert}}(n) \geqslant 2$ for every $n \geqslant 2$.

For $n = 2$ the third-order data are already carried by the second. Any triple drawn from two Bloch vectors is degenerate, so $\operatorname{Im}\Delta_3 = 0$ throughout, while $\operatorname{Re}\Delta_3$ is a function of the pairwise dot products alone. Theorem~\ref{thm:global_order} therefore separates the orbits at degree two, so $k_{\mathrm{convert}}(2) \leqslant 2$, and with the lower bound $k_{\mathrm{convert}}(2) = 2$.

For $n \geqslant 3$ degree two does not suffice. The mirror family of Sec.~\ref{sec:mirror} settles $n = 3$. For $n > 3$, append $n - 3$ copies of $\frac{\mathbb{I}}2$ to the source and to both targets. A unital map fixes $\frac{\mathbb{I}}2$, so $\boldsymbol{\sigma}^{(+)}$ is still reached by $T_+ = c\mathbb{I}_3$ while $\boldsymbol{\sigma}^{(-)}$ still forces $T_- = -c\mathbb{I}_3$, which is not completely positive for $c > \frac13$. The added degree-two invariants are $\Delta_2(\sigma_i, \frac{\mathbb{I}}2) = \frac12$ on both sides, so the padded targets remain indistinguishable at order two. Therefore, for every \(n \geqslant 3\), there are still counterexamples where the data is insufficient for the second order. Based on the above lower bound, we obtain \(k_{\mathrm{convert}}(n) = 3\).

The invariant $\Delta_2$ determines $G_{\boldsymbol{\rho}}$, and $\boldsymbol{\rho}$ is set-coherent exactly when $\rank G_{\boldsymbol{\rho}}\geqslant2$, while $\Delta_1(\rho_i) = \Tr(\rho_i) = 1$ carries no information. Thus detection is settled at degree two for every $n \geqslant 2$.
\end{proof}
This corollary partially addresses the open question in \cite{ZhangL2026} concerning the operational meaning of \(z\in\mathcal{B}_{n+1}\setminus\mathcal{B}_n\), by showing that third-order Bargmann invariants are necessary for deciding qubit multistate convertibility, whereas second-order invariants suffice for resource detection.

\section{Higher-Dimensional Rigidity: The Isotropic Collapse}
\label{sec:high_d}
For qubits, every completely positive diagonal unital channel $\Phi_{\operatorname{diag}(\boldsymbol{\lambda})}$ can be written as $\Phi_{\operatorname{diag}(\boldsymbol{\lambda})}(\rho)
=\sum_{\mu=0}^{3}p_\mu\sigma_\mu\rho\sigma_\mu$,
where $\sigma_i$ are Pauli matrices,
and the Choi eigenvalues $p_\mu$ in
Eq.~\eqref{eq:choi_spectrum} are nonnegative and sum to one. These probability mixtures fill the Pauli tetrahedron, see Fig.~\ref{fig:dimension_anomaly} (a). Its center $\boldsymbol{\lambda}=(0,0,0)$ gives
$p_\mu=1/4$ and the completely depolarizing channel. Besides the isotropic depolarizing axis
$\lambda_1=\lambda_2=\lambda_3=q$, the Pauli tetrahedron contains anisotropic channels with unequal singular values
$s_1,s_2,s_3$.

In stark contrast, for $d\geqslant3$, CPOs with non-commutative range exhibit extreme geometric rigidity. Up to unitary conjugation, with or without transposition, they scale all centered operators by a single factor $q$, as stated in the following theorem.

\begin{figure*}[t]
\includegraphics[width=\textwidth]{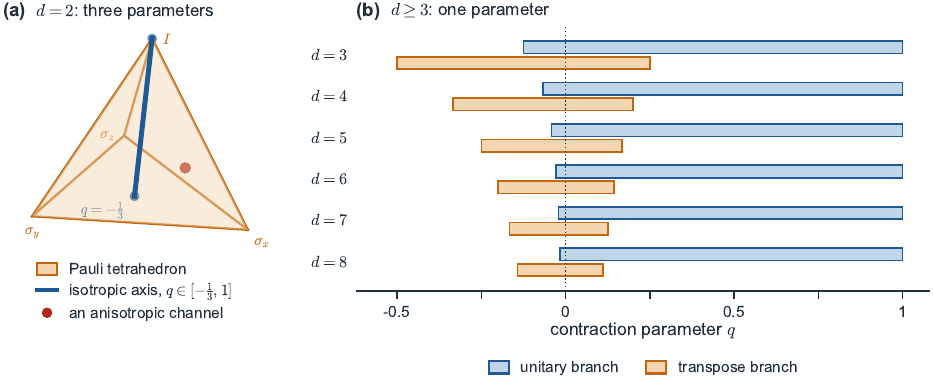}
\caption{The dimensional anomaly. (a)~For $d=2$, a unital CPO is any point of the Pauli tetrahedron, so the Bloch ball may be deformed along three independent semi-axes. The isotropic axis $\lambda_1=\lambda_2=\lambda_3=q$ is one line inside that body, $q\in[-1/3,1]$. The red dot marks a generic anisotropic channel. (b)~For $d\geqslant3$, CPOs with non-commutative range have a single scaling parameter $q$. The plot shows the admissible intervals $[-1/(d^2-1),1]$ for the unitary branch and $[-1/(d-1),1/(d+1)]$ for the transpose branch. The latter has width $2d/(d^2-1)$ and shrinks toward $q=0$ as $d$ grows. At $q=0$, the channel is completely depolarizing and has commutative range.}
\label{fig:dimension_anomaly}
\end{figure*}

\begin{theorem}[High-dimensional CPO rigidity~\cite{Hu2012,Guo2013}]
\label{thm:rigidity}
Let $\mathcal{H} \cong \mathbb{C}^d$ with $d \geqslant 3$. Suppose $\Phi \in \CPO_d$ is a non-trivial channel whose image is not commutative. Then $\Phi$ must be strictly isotropic, belonging to one of two branches:
\begin{enumerate}
    \item \textbf{Unitary-depolarizing (U-branch):}
    \begin{equation}
    \label{eq:ubranch_form}
    \Phi(X) = q\, U X U^\dagger + (1 - q)\Tr(X)\frac{\mathbb{I}}{d},
    \end{equation}
    with unitary $U \in \mathrm{U}(d)$ and contraction parameter:
    \begin{equation*}
    -\frac{1}{d^2 - 1} \leqslant q \leqslant 1, \quad q \neq 0.
    \end{equation*}
    \item \textbf{Transpose-depolarizing (T-branch):}
    \begin{equation}
    \label{eq:tbranch_form}
    \Phi(X) = q\, U X^\mathsf{T} U^\dagger + (1 - q)\Tr(X)\frac{\mathbb{I}}{d},
    \end{equation}
    with unitary $U \in \mathrm{U}(d)$ and contraction parameter:
    \begin{equation*}
    -\frac{1}{d - 1} \leqslant q \leqslant \frac{1}{d + 1}, \quad q \neq 0.
    \end{equation*}
\end{enumerate}
\end{theorem}

For $d\geqslant3$, Theorem~\ref{thm:rigidity} shows that CPO
conversion to a resourceful target admits no independent directional
scaling (Fig.~\ref{fig:dimension_anomaly}). Let
$X_i=\rho_i-\mathbb{I}/d$ and $Y_i=\sigma_i-\mathbb{I}/d$.
Such a conversion must satisfy
\begin{equation*}
    Y_i=q\,UX_iU^\dagger
    \quad\text{or}\quad
    Y_i=q\,UX_i^{\mathsf T}U^\dagger.
\end{equation*}
The detailed proof is given in Appendix~\ref{app:rigidity}.

To express this restriction in terms of invariants, define the
centered Bargmann invariants by
\[
\widetilde{\Delta}_k(\boldsymbol{\rho};i_1,\ldots,i_k)
=\Tr(X_{i_1}\cdots X_{i_k}).
\]

\begin{theorem}[Finite centered-Bargmann invariant criterion]
\label{thm:finite_invariants}
For $d \geqslant 3$, exact conversion $\boldsymbol{\rho} \succeq_{\SC} \boldsymbol{\sigma}$ between resourceful multistates holds if and only if one of the two branches of Theorem~\ref{thm:rigidity} is realized, that is, if and only if there exists $q \ne 0$ such that, for every $k \leqslant d^2$ and every index tuple, either
\begin{align}
\text{\emph{(U)}} \quad & \widetilde{\Delta}_k(\boldsymbol{\sigma}) = q^k\, \widetilde{\Delta}_k(\boldsymbol{\rho}),
&& -\tfrac{1}{d^2-1} \leqslant q \leqslant 1, \label{eq:ubranch_inv} \\
\text{\emph{(T)}} \quad & \widetilde{\Delta}_k(\boldsymbol{\sigma}) = q^k\, \overline{\widetilde{\Delta}_k(\boldsymbol{\rho})},
&& -\tfrac{1}{d-1} \leqslant q \leqslant \tfrac{1}{d+1}. \label{eq:tbranch_inv}
\end{align}
The degree bound $d^2$ is the Razmyslov--Procesi bound for the trace algebra of $d \times d$ matrices under simultaneous conjugation~\cite{Razmyslov1974,Procesi1976}.
\end{theorem}

\begin{proof}
Write $Y_i = \sigma_i - \frac{\mathbb{I}}d$. By Theorem~\ref{thm:rigidity} a resourceful conversion is realized by $Y_i = q\,U X_i U^\dagger$ with $-\tfrac{1}{d^2-1} \leqslant q \leqslant 1$, or by $Y_i = q\,U X_i^\mathsf{T} U^\dagger$ with $-\tfrac{1}{d-1} \leqslant q \leqslant \tfrac{1}{d+1}$, in both cases with $q \neq 0$ and $U$ unitary.

Necessity. On the unitary branch the conjugations cancel cyclically,
\begin{equation*}
\Tr[Y_{i_1} \cdots Y_{i_k}] = q^k \Tr[X_{i_1} \cdots X_{i_k}],
\end{equation*}
which is~\eqref{eq:ubranch_inv}. On the transpose branch the same cancellation leaves a transposed product, and since $\Tr[A^\mathsf{T}] = \Tr[A]$ and the $X_i$ are Hermitian,
\begin{align*}
\Tr[X_{i_1}^\mathsf{T} \cdots X_{i_k}^\mathsf{T}]
  &= \Tr[(X_{i_k} \cdots X_{i_1})^\mathsf{T}] \notag \\
  &= \Tr[X_{i_k} \cdots X_{i_1}] \notag \\
  &= \overline{\Tr[X_{i_1} \cdots X_{i_k}]},
\end{align*}
which is~\eqref{eq:tbranch_inv}.

Sufficiency. Suppose~\eqref{eq:ubranch_inv} holds for every word
of degree at most $d^2$. By the Razmyslov--Procesi generation
bound, the Hermitian tuples $(Y_i)$ and $(qX_i)$ then have the
same trace invariants of every degree~\cite{Razmyslov1974,Procesi1976}.

Polynomial invariants distinguish distinct closed orbits under simultaneous $\mathrm{GL}(d,\mathbb C)$ conjugation. Since $(Y_i)$ and $(qX_i)$ are Hermitian tuples, their orbits are closed. Their equal invariants therefore place them in the same orbit, so there exists a single $S\in\mathrm{GL}(d,\mathbb C)$ such that $Y_i=S(qX_i)S^{-1}$ for every $i$.

The conjugating element may then be taken unitary. Both tuples are Hermitian, so taking adjoints in $Y_i = S(q X_i) S^{-1}$ gives $(S^\dagger S) X_i = X_i (S^\dagger S)$ for every $i$. Write the polar decomposition $S = UP$ with $U$ unitary and $P = (S^\dagger S)^{1/2}$. Since $P$ is a function of $S^\dagger S$, it also commutes with every $X_i$. Therefore,
\begin{equation*}
Y_i = U P (q X_i) P^{-1} U^\dagger = U (q X_i) U^\dagger .
\end{equation*}
Since $q$ lies in the stated admissible interval, this map is a CPO and realizes the required conversion. The transpose branch is identical after replacing $(X_i)$ by $(X_i^\mathsf{T})$, which is again a Hermitian tuple with conjugated invariants.
\end{proof}

The complex conjugation in the transpose branch is the higher-dimensional counterpart of the qubit mirror obstruction. Because the $X_i$ are Hermitian, $\Tr[X_{i_1}^\mathsf{T} \cdots X_{i_k}^\mathsf{T}] = \Tr[X_{i_k} \cdots X_{i_1}] = \overline{\widetilde{\Delta}_k(\boldsymbol{\rho})}$: reversing the index order conjugates the invariant. The two branches therefore coincide for $k \leqslant 2$, and $k = 3$ is the earliest order at which they can differ, namely when some cubic invariant has non-zero imaginary part: orientation is invisible to pairwise data in every dimension, not only for qubits.

\section{Arbitrary Multistate Sizes and Degenerate Geometries}
\label{sec:arbitrary_n}

Corollary~\ref{cor:global_order} settles abstractly that no invariant above degree three is ever needed for qubit conversion. We now turn that statement into an explicit criterion for multistates of arbitrary size.

For qubit $n$-multistates, let
$G_{\boldsymbol{\rho}},G_{\boldsymbol{\sigma}}$
denote their full $n\times n$ Bloch Gram matrices.
Set $J=\{1,2,3\}$, write
$R_J=[\mathbf r_1,\mathbf r_2,\mathbf r_3]$ and
$S_J=[\mathbf s_1,\mathbf s_2,\mathbf s_3]$,
and use $G[I,K]$ for a submatrix.
When $G_{\boldsymbol{\rho}}[J,J]$ is invertible, define
\[
C_{\boldsymbol{\rho}}[:]
=G_{\boldsymbol{\rho}}[:,J]\,
G_{\boldsymbol{\rho}}[J,J]^{-1}
\in\mathbb R^{n\times3}.
\]

\begin{theorem}[Explicit consistency criterion for spanning multistates]
    Let $\boldsymbol{\rho}$ and $\boldsymbol{\sigma}$ be qubit
    $n$-multistates, where $n\geqslant3$ and
    $\boldsymbol{\sigma}$ is resourceful.
    Suppose that the Bloch vectors
    $\mathbf r_1,\mathbf r_2,\mathbf r_3$ span $\mathbb{R}^3$. Then $\boldsymbol{\rho}\succeq_{\SC}\boldsymbol{\sigma}$ if and only if the candidate map with Bloch matrix $T=S_JR_J^{-1}$ is completely positive and
    for every $k>3$,
    \begin{align}
        G_{\boldsymbol{\sigma}}[J,k]
        &=G_{\boldsymbol{\sigma}}[J,J]\,
        C_{\boldsymbol{\rho}}[k]^{\mathsf T},
        \label{eq:extra_state_dir}\\
        G_{\boldsymbol{\sigma}}[k,k]
        &=C_{\boldsymbol{\rho}}[k]\,
        G_{\boldsymbol{\sigma}}[J,J]\,
        C_{\boldsymbol{\rho}}[k]^{\mathsf T}.
        \label{eq:extra_state_norm}
    \end{align}
    Here $G[I,K]$ denotes the submatrix with rows in $I$
    and columns in $K$.
\end{theorem}

The two consistency conditions involve only second-order Bargmann invariants. Thus, beyond the spanning triple, no third- or higher-order invariants involving the additional states are needed.


\begin{proof}
Since \(\boldsymbol{\rho}\succeq_{\SC}\boldsymbol{\sigma}\), the restriction to the first three states gives \((\rho_1,\rho_2,\rho_3)\succeq_{\SC}(\sigma_1,\sigma_2,\sigma_3)\). The realizing CPO is unital because \(\boldsymbol{\sigma}\) is resourceful. Hence its Bloch matrix satisfies \(TR_J=S_J\), and therefore \(T=S_JR_J^{-1}\).

For every $k>3$, we have $\mathbf{s}_k=T\mathbf{r}_k$. From $R_J^{\mathsf T}R_J=G_{\boldsymbol{\rho}}[J,J]$ and $R_J^{\mathsf T}\mathbf{r}_k=G_{\boldsymbol{\rho}}[J,k]$, we obtain $R_J^{-1}\mathbf{r}_k = G_{\boldsymbol{\rho}}[J,J]^{-1}
G_{\boldsymbol{\rho}}[J,k] =C_{\boldsymbol{\rho}}[k]^{\mathsf T}$
and $T\mathbf{r}_k=S_JR_J^{-1}\mathbf{r}_k
=S_JC_{\boldsymbol{\rho}}[k]^{\mathsf T}$.
Therefore,
\begin{eqnarray*}
    G_{\boldsymbol\sigma}[J,k]
    =S_J^{\mathsf T}\mathbf s_k=S_J^{\mathsf T}S_JR_J^{-1}\mathbf r_k=G_{\boldsymbol\sigma}[J,J]\,C_{\boldsymbol\rho}[k]^{\mathsf T}.
\end{eqnarray*}
And taking the squared norm gives~\eqref{eq:extra_state_norm}.
Thus both conditions hold.

For sufficiency, suppose the candidate unital map with Bloch matrix $T=S_JR_J^{-1}$ is completely positive and
that~\eqref{eq:extra_state_dir} and
\eqref{eq:extra_state_norm} hold for every $k>3$.
Since $TR_J=S_J$, this map sends the first three source
states to their respective targets.

Fix $k>3$ and write
$a_k=C_{\boldsymbol{\rho}}[k]^{\mathsf T}$.
The definition of $C_{\boldsymbol{\rho}}$ gives
$R_J^{-1}\mathbf r_k=a_k$, and hence $T\mathbf r_k=S_Ja_k\in\operatorname{range}(S_J)$. Set $\mathbf w=\mathbf s_k-T\mathbf r_k$. Using~\eqref{eq:extra_state_dir}, we obtain
\[
\begin{aligned}
    S_J^{\mathsf T}\mathbf w
    &=S_J^{\mathsf T}\mathbf s_k
    -S_J^{\mathsf T}S_Ja_k\\
    &=G_{\boldsymbol{\sigma}}[J,k]
    -G_{\boldsymbol{\sigma}}[J,J]a_k
    =0.
\end{aligned}
\]
Thus $\mathbf w\perp\operatorname{range}(S_J)$.
Since $T\mathbf r_k$ belongs to this range,
$\mathbf s_k=T\mathbf r_k+\mathbf w$ is an orthogonal
decomposition, so
\[
\|\mathbf s_k\|^2
=\|T\mathbf r_k\|^2+\|\mathbf w\|^2.
\]
Meanwhile,~\eqref{eq:extra_state_norm} gives
\[
\|\mathbf s_k\|^2
=G_{\boldsymbol{\sigma}}[k,k]
=a_k^{\mathsf T}G_{\boldsymbol{\sigma}}[J,J]a_k
=\|S_Ja_k\|^2
=\|T\mathbf r_k\|^2.
\]
Therefore $\|\mathbf w\|^2=0$, and $\mathbf s_k=T\mathbf r_k$.

This holds for every $k>3$, so the same channel converts the entire multistate. As a unital qubit channel, it preserves commutativity and
is therefore a CPO. The argument does not require $S_J$ to be invertible.
\end{proof}

For degenerate geometries where $\mathrm{span}\{\mathbf{r}_i\}$ has dimension 2 (planar) or dimension 1 (collinear):
\begin{itemize}
\item \textbf{Planar sector ($\operatorname{rank}=2$):}
For a resourceful target, any realizing CPO must be unital. Its Bloch action is therefore linear. A map satisfying $T\mathbf r_i=\mathbf s_i$ for every $i$ exists if and only if
\begin{equation}
    \label{eq:planar_consistency}
    \ker R\subseteq\ker S.
\end{equation}
Equivalently, every linear relation among the source Bloch
vectors must also hold among the target Bloch vectors.

Let $a \geqslant b \geqslant 0$ be the singular values of $T$ restricted to $\operatorname{span}\{\mathbf{r}_i\}$, and let $\tau \in \mathbb{R}$ be the signed coefficient in action \(T\mathbf n=\tau\mathbf n'\). Rotations in the source and target planes bring the constrained part to $\diag(a, b)$, so the extension has signed spectrum $\boldsymbol{\lambda} = (a, b, \tau)$ and complete positivity is~\eqref{eq:choi_spectrum} evaluated there. Of the four eigenvalues, $1 - a - b + \tau \geqslant 0$ and $1 - a + b - \tau \geqslant 0$ bind, which leaves an interval feasibility problem on $\tau$,
    \begin{equation*}
    a + b - 1 \ \leqslant\ \tau \ \leqslant\ 1 - a + b,
    \end{equation*}
    an interval that is non-empty exactly when $a \leqslant 1$. Here $a + b > 1$ is no obstruction: the free normal direction absorbs it, and when $a = b$ the isotropic extension $\tau = a$ stays admissible up to $a = 1$. Thus, a planar source can be converted to a resourceful target
    by a CPO if and only if $\ker R\subseteq\ker S$ and $a\leqslant1$.
    Free targets are covered by the preceding POVM result.
\item \textbf{Collinear sector ($\operatorname{rank} \leqslant 1$):} The states commute mutually, meaning the initial ensemble is free. Because CPO operations preserve commutativity, conversion into any resourceful target is strictly impossible.
\end{itemize}

\begin{table*}[t]
\centering
\caption{Closed-form content of the conversion criterion. \emph{Top:} the three qubit sectors, with $s_1 \geqslant s_2 \geqslant s_3$ the singular values of $T = SR^{-1}$, $\eta=\operatorname{sgn}\Big(\frac{\operatorname{Im}\Delta_3(\boldsymbol{\sigma})}{\operatorname{Im}\Delta_3(\boldsymbol{\rho})}\Big)$ its orientation, and shares taken in $\mathcal{S} = \{(s_1,s_2,s_3)~|~1 \geqslant s_1 \geqslant s_2 \geqslant s_3 \geqslant 0\}$. The two orientations disagree on $\mathcal{C}_{\mathrm{chir}}$ alone. \emph{Bottom:} the two rigid branches for $d \geqslant 3$, which differ by the conjugation in the transpose row.}
\label{tab:benchmarks}
\renewcommand{\arraystretch}{1.25}
\setlength{\tabcolsep}{8pt}
\sbox{\tabQ}{\begin{tabular}{@{}llccr@{}}
\toprule
\textbf{Sector} & \textbf{Condition} & $\eta = +1$ & $\eta = -1$ & \textbf{Share of $\mathcal{S}$} \\
\midrule
I & $s_1 + s_2 + s_3 \leqslant 1$ & \checkmark & \checkmark & $\frac16$ \\
$\mathcal{C}_{\mathrm{chir}}$ & $s_1 + s_2 + s_3 > 1 \geqslant s_1 + s_2 - s_3$ & \checkmark & $\boldsymbol{\times}$ & $\frac13$ \\
II & $s_1 + s_2 - s_3 > 1$ & $\boldsymbol{\times}$ & $\boldsymbol{\times}$ & $\frac12$ \\
\bottomrule
\end{tabular}}
\usebox{\tabQ}

\vspace{2.2ex}

\begin{tabular*}{\wd\tabQ}{@{\extracolsep{\fill}}lll@{}}
\toprule
\textbf{Branch} ($d \geqslant 3$) & \textbf{Invariant relation} & \textbf{Admissible $q$} \\
\midrule
U, unitary & $\widetilde{\Delta}_k(\boldsymbol{\sigma}) = q^k\, \widetilde{\Delta}_k(\boldsymbol{\rho})$ & $-\frac1{d^{2}-1}\leqslant q \leqslant 1$ \\
T, transpose & $\widetilde{\Delta}_k(\boldsymbol{\sigma}) = q^k\, \overline{\,\widetilde{\Delta}_k(\boldsymbol{\rho})\,}$ & $-\frac1{d-1} \leqslant q \leqslant \frac1{d+1}$ \\
\bottomrule
\end{tabular*}
\end{table*}

\section{Volume of the Chirality Shell}
\label{sec:benchmarks}

We now quantify how often chirality alone decides convertibility (Table~\ref{tab:benchmarks}). The three regions of Sec.~\ref{sec:qubit_conversion} partition the ordered singular-value simplex $\mathcal{S} = \{ (s_1, s_2, s_3) : 1 \geqslant s_1 \geqslant s_2 \geqslant s_3 \geqslant 0 \}$, whose volume is $\mathrm{Vol}(\mathcal{S}) = \frac16$. Their shares are rational and can be evaluated in closed form.

\begin{proposition}[Exact volume shares]
\label{prop:volumes}
With respect to the uniform measure on $\mathcal{S}$,
\begin{equation*}
\frac{\mathrm{Vol}(\mathrm{I})}{\mathrm{Vol}(\mathcal{S})} = \frac{1}{6}, \quad
\frac{\mathrm{Vol}(\mathrm{II})}{\mathrm{Vol}(\mathcal{S})} = \frac{1}{2}, \quad
\frac{\mathrm{Vol}(\mathcal{C}_{\mathrm{chir}})}{\mathrm{Vol}(\mathcal{S})} = \frac{1}{3}.
\end{equation*}
\end{proposition}

\begin{proof}
Let $D=\{(s_1,s_2,s_3):s_i\geqslant0,\ s_1+s_2+s_3\leqslant1\}$. This is the standard simplex, with volume $1/6$. Region I is the part of $D$ satisfying $s_1\geqslant s_2\geqslant s_3$. Permuting the three coordinates divides $D$ into six parts of equal volume. Therefore, the volume of Region I is
$$
\operatorname{Vol}(\mathrm I)=\frac{\operatorname{Vol}(D)}{6}
=\frac{1}{36}.
$$

For Region~II, the constraint $s_3 < s_1 + s_2 - 1$ together with $s_1 \leqslant 1$ gives $s_3 < s_1+s_2-1 \leqslant s_2$, so
\begin{align*}
\mathrm{Vol}(\mathrm{II}) &= \int_0^1\!\!\int_0^{s_1} (s_1+s_2-1)_+ {\rm d}s_2{\rm d}s_1 \\
&= \int_{\frac12}^{1} \frac{(2s_1-1)^2}{2}{\rm d}s_1 = \frac{1}{12}.
\end{align*}
The shell is the complement, $\mathrm{Vol}(\mathcal{C}_{\mathrm{chir}}) = \frac16 - \frac1{36} - \frac1{12} = \frac1{18}$. Dividing by $\mathrm{Vol}(\mathcal{S}) = \frac16$ gives the stated shares.
\end{proof}

A transformation is realizable in the orientation-preserving direction exactly on $\mathrm{I} \cup \mathcal{C}_{\mathrm{chir}}$, which fills half of $\mathcal{S}$. Reflecting it obstructs it exactly on $\mathcal{C}_{\mathrm{chir}}$, so
\begin{equation*}
P_{\mathrm{obs}} = \frac{\mathrm{Vol}(\mathcal{C}_{\mathrm{chir}})}{\mathrm{Vol}(\mathrm{I} \cup \mathcal{C}_{\mathrm{chir}})} = \frac{1/3}{1/2} = \frac{2}{3}.
\end{equation*}
Two thirds of the orientation-preserving admissible region, by relative Lebesgue volume, becomes impossible under reflection, although reflection preserves every pairwise inner product and distance.

\section{Discussion}
\label{sec:discussion}

Detecting set coherence and deciding its conversion can require different information. For qubits, detection uses unoriented pairwise data, whereas conversion can depend on the orientation recorded by a third-order Bargmann invariant.

This orientation dependence arises from complete positivity. A quantum channel must have a positive semidefinite Choi matrix, which restricts the admissible Bloch transformations. For
diagonal unital qubit maps, the admissible region is the Pauli tetrahedron, whose vertices correspond to the four Pauli unitary channels. The universal inversion $\mathbf r\mapsto-\mathbf r$ is unitarily equivalent to transposition, which acts as $(x,y,z)\mapsto(x,-y,z)$. Both maps are positive but not completely positive. Consequently, a mirror target can have the same pairwise invariants as a reachable target while failing the complete-positivity criterion.

Qubits and qudits with $d\geqslant3$ differ in this respect. In higher dimensions, CPOs with non-commutative range have only one scaling parameter, up to unitary conjugation or transposition. The anisotropic transformations available to qubit CPOs are therefore a low-dimensional exception within this class of operations.

Where the classical boundary falls for an individual state or process is under active revision~\cite{Cobucci2026,ZhangS2026,ZhangY2026}. A multistate carries an orientation that no single-state quantifier resolves, and the results here locate it at degree three.

The criterion also bears on the sensitivity of coherence to dephasing and on the trade-off between coherence and mixedness under noise~\cite{Yadin2019,Che2023}. Noise contracts states in state space, and Theorem~\ref{thm:qubit_exact} says which contracted multistate configurations stay reachable without creating relational coherence that was not there to begin with.

\section{Conclusion}
\label{sec:conclusion}

We have developed an exact resource-conversion theory for
quantum multistates under common commutativity-preserving
operations. Three results follow.

\begin{enumerate}
    \item \textbf{Detection and conversion require different
        orders.} For qubit multistates,
    $k_{\mathrm{detect}}(n)=2$ for every $n\geqslant2$.
    Conversion requires order two for pairs and order three
    for every $n\geqslant3$.

    \item \textbf{The chirality-sensitive shell.}
    On $\mathcal{C}_{\mathrm{chir}}$, orientation alone
    distinguishes realizable from forbidden conversions
    with the same singular values. The shell occupies one
    third of the ordered singular-value simplex and two
    thirds of the volume admissible at $\eta=+1$.

    \item \textbf{Rigidity above $d=2$.}
    For $d\geqslant3$, conversions to resourceful targets
    follow an isotropic depolarizing or
    transpose-depolarizing branch. Centered Bargmann
    invariants of degree at most $d^2$ decide
    convertibility.
\end{enumerate}

The same question can be studied for other classes of free
operations and for multistates distributed between parties,
where orientation data need not reduce to a single scalar sign.

\begin{acknowledgments}
The authors designed the study, carried out the mathematical analysis, and prepared the manuscript. Generative-AI tools were used for figure preparation and language refinement. All mathematical statements were independently reviewed and verified by the authors, who take full responsibility for the content.
\end{acknowledgments}

\section*{Data availability}
No data were created or analyzed in this study.


\appendix

%
%

\section{Proof of Theorem~\ref{thm:rigidity}}
\label{app:rigidity}

The two-branch classification is due to Hu, Fan, Zhou and Liu for $d = 3$ and to Guo and Hou for all $d \geqslant 3$~\cite{Hu2012,Guo2013}, and we take it as given: a commutativity-preserving channel on $\mathbb{C}^d$ with $d \geqslant 3$ whose range is non-commutative has one of the two forms~\eqref{eq:ubranch_form} and~\eqref{eq:tbranch_form}, with $U$ unitary and $q$ real. Both are trace-preserving by construction, so what remains is to read the admissible $q$ off complete positivity. This also fixes the normalisation used in Sec.~\ref{sec:high_d}.

The Choi-Jamio{\l}kowski matrix $J_\Phi = (\Phi \otimes \mathrm{id})(|\Omega\rangle\langle\Omega|)$, where $|\Omega\rangle = \tfrac{1}{\sqrt{d}}\sum_{i=1}^d |i i\rangle$, has eigenvalues that depend linearly on $q$:
\begin{itemize}
    \item For the unitary branch (U-branch), $J_\Phi$ has an eigenvalue of multiplicity 1 equal to $\frac{1 + (d^2-1)q}{d^2}$, and an eigenvalue of multiplicity $d^2-1$ equal to $\frac{1-q}{d^2}$. Enforcing $J_\Phi \geqslant 0$ yields $-\frac{1}{d^2-1} \leqslant q \leqslant 1$.
    \item For the transpose branch (T-branch), $J_\Phi$ acts as the depolarized partial transpose, with eigenvalue $\frac{1 + (d-1)q}{d^2}$ of multiplicity $\frac{d(d+1)}2$ on the symmetric subspace and eigenvalue $\frac{1 - (d+1)q}{d^2}$ of multiplicity $\frac{d(d-1)}2$ on the antisymmetric subspace. Enforcing $J_\Phi \geqslant 0$ yields $-\frac{1}{d-1} \leqslant q \leqslant \frac{1}{d+1}$.
\end{itemize}
This completes the proof. \qed


\end{document}